\documentclass[notitlepage,aps,pra,twocolumn,groupaddress,10pt]{revtex4-1}

\usepackage[english]{babel}

\usepackage{graphicx}
\usepackage[colorlinks=true, allcolors=blue]{hyperref}

\usepackage{stmaryrd}
\usepackage{amssymb,amsmath,amsthm,amsfonts,amsbsy}
\usepackage{bm,bibunits,color,chngcntr,epsfig,epstopdf,graphicx,dsfont}
\usepackage{hyperref,lipsum,,makecell,mathrsfs,rotating,setspace}
\usepackage[english]{babel}
\usepackage[normalem]{ulem}

\newcommand{\be}{\begin{equation}}
\newcommand{\ee}{\end{equation}}
\newcommand{\bea}{\begin{eqnarray}}
\newcommand{\eea}{\end{eqnarray}}
\newcommand{\ket}{\rangle}

\newcommand{\I}{\mathds{1}}
\newcommand{\ra}{\rightarrow}

\newcommand{\ba}{\begin{align}}
\newcommand{\ea}{\end{align}}

\usepackage{tikz}
\usetikzlibrary{shapes,arrows,positioning,fit,calc}

\renewcommand{\ket}[1]{|#1\rangle}

\begin{document}

\title{Quantum conditional mutual information and channel capacity}
\author{Dong-Sheng Wang}
\affiliation{Institute of Theoretical Physics, Chinese Academy of Sciences, Beijing 100190, China \\
School of Physical Sciences, University of Chinese Academy of Sciences, Beijing 100049, China}

\newtheorem{theorem}{Theorem}
\newtheorem{prop}[theorem]{Proposition}
\newtheorem{corollary}[theorem]{Corollary}
\newtheorem{open problem}[theorem]{Open Problem}
\newtheorem{conjecture}[theorem]{Conjecture}
\newtheorem{definition}{Definition}
\newtheorem{remark}{Remark}
\newtheorem{example}{Example}
\newtheorem{task}{Task}
\newtheorem{protocol}{Protocol}

\begin{abstract}
Information measures acquire operational meaning through coding theorems. The quantum conditional mutual information (QCMI) is nonnegative due to strong subadditivity, yet a direct connection with channel coding has remained elusive. In this work, we propose a quantum communication task—conditional quantum communication—that fills this gap. We show that the optimal rate for establishing quantum correlation between two parties, assisted by a third system, is given by half the QCMI. This result naturally extends the classical key generation capacity of Csiszár and Ahlswede to the quantum domain. We place our model within the family tree of quantum protocols and compute the conditional capacity for several example channels. Our results provide new insights for code design in reliable quantum information processing.
\end{abstract}
\date{\today}

\maketitle

\begin{spacing}{1.0}

\section{Introduction}

Information measures lie at the heart of information theory, 
and they acquire operational meaning through Shannon's coding theorems~\cite{Wat18,Wil17,Hay17}.
Entropy and mutual information are two seminal examples.
Another such measure is conditional mutual information.
In the quantum case, a celebrated result is the strong subadditivity of von Neumann entropy,
which implies that the quantum conditional mutual information (QCMI)
is nonnegative~\cite{LR73}.

Quantum coding models can be broadly classified into three types: 
source coding, channel coding, and state conversion.
The latter two are related via the state-channel duality,
originally formed in the classical setting and flourished in the quantum case 
by Devetak and coauthors~\cite{Dev05d},  
extending Choi's isomorphism~\cite{Cho75}.
In these models, the entropy $S(A)$ bounds the compression rate for source coding,
while the mutual information $I(A:B)$ between input $A$ and output $B$
bounds the coding rate through a channel.
The classical conditional entropy $S(A|B)$ bounds the rate for conditional 
source coding, as shown by Slepian and Wolf~\cite{CT06}.
Its quantum version, known as state merging~\cite{HOW07}, reveals a key quantum feature: 
the quantum conditional entropy can be negative,
and its negative value—the coherent information—serves as the rate for generating ebits.
This coherent information also characterizes the standard quantum channel capacity~\cite{SN96,Llo97,BNS98,Dev05}
and its dual task, entanglement distillation~\cite{DW05}.

To make the nonnegative quantum mutual information (QMI) a capacity measure, 
one can use entanglement-assisted (EA) models.
These yield the EA channel capacities~\cite{BSS+99,BSST02,BDH+14}
and coherent merging, leading to the family tree of quantum coding protocols~\cite{DHW08,ADHW09,DHW04,DH11}.
Recently, state-adaptive (SA) coding models~\cite{Wang26} showed that the SA capacity is equivalent to the EA capacity.
This highlights an intriguing relation: a known quantum state admits a classical description,
and the quantum capacity becomes half of the corresponding classical capacity in the EA or SA setting.

In this work, we focus on the conditional mutual information (CMI).
Classically, CMI appears in many tasks, such as broadcast channels,
channels with side information, 
and key generation~\cite{CT06,CK11}.
In the quantum case, QCMI has found operational meaning in state redistribution~\cite{LD06,DY08,YD09,Opp08,DHO16,YBW08},
aligning with source coding and state merging,
quantum broadcast channels~\cite{YHD11,QSW18},
and also in other tasks~\cite{PEW17,BBM+18,Pereg19,SWW20}. 
However, a direct relation with channel coding has been missing: 
specifically, a capacity given by $\max I(A:B|C)$ with $A$ as the input, $B$ as the output,
and $C$ as a conditional quantum system.
In this work, we solve this problem by introducing a model we call \emph{conditional quantum communication}.
We also analyze its relation to known tasks such as entanglement swapping and quantum repeaters~\cite{BDCZ98}. 

Our model closely follows the key generation models of Csisz\'ar and others~\cite{AC93,AC98,CK11}. 
In those classical models, the goal is to generate a key between $A$ and $B$
while keeping $C$ decoupled.
In the original formulation, a broadcast channel $A \to B \times C$ is considered,
and the variable $C$ is assumed to be known to $A$ or $B$ 
in order to achieve the CMI capacity.
We develop the quantum version of this setting and 
also analyze its relation to private capacity for wiretap channels~\cite{Wyner75}.

The remainder of the paper is organized as follows. 
Section~\ref{sec:prelim} reviews the definition and properties of QCMI.
Section~\ref{sec:model} defines our model 
and presents the main theorem with proof. 
Section~\ref{sec:relations} connects the conditional capacity to other capacities 
in the protocol family tree. 
Section~\ref{sec:examples} evaluates the conditional capacity for concrete channels. 
Section~\ref{sec:conclusion} concludes with perspectives on code design 
and applications.

\section{Preliminaries}\label{sec:prelim}

\subsection{Quantum conditional mutual information}

In this work, we use $A$, $B$, $C$ as labels for quantum or classical systems,
and $S(A)$ denotes the von Neumann entropy of system $A$, regardless of its classical or quantum nature.
The context will make clear whether a system is classical or quantum.

For a bipartite quantum state $\rho_{AB}$, 
the quantum mutual information (QMI) is defined as
\begin{equation}
I(A:B) = S(A) + S(B) - S(AB).
\end{equation}
The coherent information is $I_c(A\rangle B)=S(B)-S(AB)=-S(A|B)$, 
which can be negative.

For a tripartite quantum state $\rho_{ABC}$, 
the quantum conditional mutual information (QCMI) is defined as
\begin{equation}
I(A:B|C) = S(AC) + S(BC) - S(C) - S(ABC).
\end{equation}
A fundamental property is strong subadditivity, which states that QCMI is nonnegative:
\begin{equation}
I(A:B|C) \ge 0.
\end{equation}
Equality holds if and only if the state is a quantum Markov chain~\cite{HJPW04}.
Such a state takes the form
\begin{equation}
\rho_{ABC} = \bigoplus_j q_j \; \rho_{A c_j^L}^{(j)} \otimes \rho_{c_j^R B}^{(j)},
\end{equation}
with a decomposition of the $C$ system as
$\mathcal{H}_C = \bigoplus_j \mathcal{H}_{c_j^L} \otimes \mathcal{H}_{c_j^R}$,
where $\{q_j\}$ is a probability distribution, 
$\rho_{A c_j^L}^{(j)}$ is a state on $\mathcal{H}_A \otimes \mathcal{H}_{c_j^L}$, 
and $\rho_{c_j^R B}^{(j)}$ is a state on $\mathcal{H}_{c_j^R} \otimes \mathcal{H}_B$.
In other words, conditioned on a classical label $j$ obtained by measuring the $C$ system in orthogonal subspaces, the systems $A$ and $B$ become independent.

A deep connection exists between QCMI and the Petz recovery map \cite{Pet86,OP93,SW15}. 
Given a state $\rho$ and a channel $\Phi$,
the Petz map $\mathcal{R}_{\rho,\Phi}$ is defined by
\begin{equation}
\mathcal{R}_{\rho,\Phi}(\sigma) = \rho^{1/2} \Phi^\dagger\bigl(\Phi(\rho)^{-1/2} \sigma \Phi(\rho)^{-1/2}\bigr) \rho^{1/2}.
\end{equation}
For any tripartite state $\rho_{ABC}$, there exists a Petz map $\mathcal{R}_{C\to AC}$ 
such that
\begin{equation}
F\bigl(\rho_{ABC},\; \mathcal{R}_{C\to AC}(\rho_{BC})\bigr) \ge 2^{-\frac12 I(A:B|C)},
\end{equation}
where $F(\rho,\sigma)=\|\sqrt{\rho}\sqrt{\sigma}\|_1^2$ is the fidelity. 
In a communication setting, the Petz map can be used by $C$ (potentially an eavesdropper) to extract information from $AB$.
Thus, to ensure security, the QCMI $I(A:B|C)$ needs to be maximized.

For multipartite systems, QCMI obeys a subadditivity property under composition:
\begin{equation}
I(A_1A_2:B_1B_2|C_1C_2) \le I(A_1:B_1|C_1) + I(A_2:B_2|C_2),
\end{equation}
which follows from strong subadditivity and the chain rule for mutual information.
It also satisfies the data-processing inequality (DPI):
\begin{equation}
I(A:B|C) \ge I(A:\mathcal{E}(B)|C)
\end{equation}
for any quantum channel $\mathcal{E}$ on system $B$ (or $A$).
Moreover, QCMI is continuous under small perturbations of the state,
satisfying Fannes-type inequalities.
These properties will be used in the coding theorem.

\subsection{Examples of QCMI}

\begin{figure}
    \centering
    \includegraphics[width=0.35\textwidth]{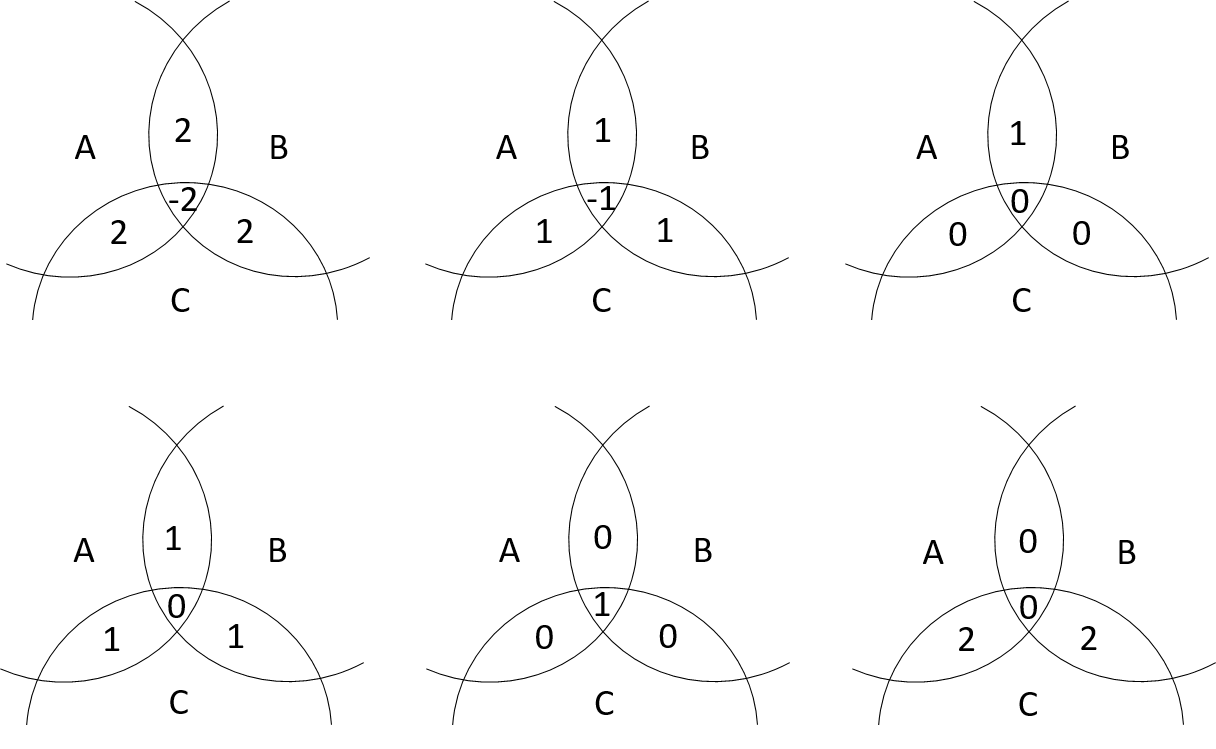}
    \caption{Venn diagram for the mutual information of 
    the six states considered in this paper: $\psi_1$ to $\psi_6$
    (from left to right, top to bottom).
    }
    \label{fig:state}
\end{figure}

We consider six paradigmatic tripartite states
and analyze their QCMI,
as well as the ordinary QMI and the tripartite information 
$I_3(A:B:C)=I(A:B)-I(A:B|C)$.  
All systems are taken to be qubits ($d=2$) except 
where a higher-dimensional register is required to label an orthogonal set of Bell states.  
There are pure and mixed states, yet for simplicity we use $\psi$ to denote a state.
The six states are defined as
\bea 
\psi_1 & =& \frac{1}{4} \sum_{i=1}^{4} |\omega_i\rangle\langle\omega_i|_{AB} \otimes |i\rangle\langle i|_C , \\ 
\psi_2 & =& \frac{1}{2}\bigl( \rho_{\Phi} \otimes |0\rangle\langle0|_C + \rho_{\Psi} \otimes |1\rangle\langle1|_C \bigr),\\
\psi_3 & =& \rho_{\Phi} \otimes |0\rangle\langle0|_C ,\\
\psi_4 & =& \frac{1}{\sqrt{2}} \bigl( |000\rangle + |111\rangle \bigr),\\
\psi_5 & =& \frac{1}{2}\bigl( |000\rangle\langle000| + |111\rangle\langle111| \bigr),\\
\psi_6 & =& \frac{1}{2} \sum_{i=1}^{4} |\omega_i\rangle_{AB} \otimes |i\rangle_C.
\eea
for $\{|\omega_i\rangle\}$ as the four orthogonal Bell states, 
$\rho_{\Phi} = \frac{1}{2}(|00\rangle\langle00|+|11\rangle\langle11|)$ and
$\rho_{\Psi} = \frac{1}{2}(|01\rangle\langle01|+|10\rangle\langle10|)$ as decohered Bell states,
`debits' (also known as common randomness).
The values are shown in Fig.~\ref{fig:state}. 

We see that $\psi_{1,2,3,4}$ have positive $I(A:B|C)$.
The state $\psi_1$ is optimal and it is the state used in entanglement swapping after the Bell measurement.
The $\psi_6$ is the state before the measurement, yet it has zero $I(A:B|C)$. 
For it $C$ is maximally entangled with $AB$ so $I(A:C)=I(B:C)=2$,
and it only has 2-body correlation.
It needs to measure $C$ in the right basis in order to extract an ebit on $AB$.
On the contrary, all 2-body correlation of $\psi_1$ is zero, 
and its $I_3(A:B:C)$ is optimal. 
Moreover, no matter how many Bell states are mixed, the $I(A:B|C)$ in it is always 2. 

The state $\psi_2$ is a mixture of debits,
and  $\psi_3$ is a single debit. 
Yet $\psi_2$ does not have 2-body correlation and it has 3-body $I_3(A:B:C)$,
with values half of those in $\psi_1$. 
The state $\psi_1$ resp. $\psi_2$ is the ideal state of conditional ebits resp. debits. 

The state $\psi_4$ is a GHZ state, 
and it has 2-body correlation without 3-body one. 
Actually, for any pure state, $I_3=0$, 
which means it is not a measure of genuine three-way entanglement.
The $X$ basis measurement on GHZ can be used to extract an ebit on $AB$. 
The state $\psi_5$, as the decohered GHZ state,
turns the quantum correlations into a classical shared random bit: 
$I(A:B)=I(A:C)=I(B:C)=1$, but $I(A:B|C)=0$, 
leading to $I_3=1$-the only case with strictly positive $I_3$, 
signifying one bit of common randomness shared among all three parties. 
It is also interesting to compare $\psi_5$ with $\psi_2$,
with the latter sharing a private key against $C$.

\section{Conditional quantum communication}\label{sec:model}

\subsection{Csisz\'ar's  key generation protocol}

In this section, we define the model for conditional quantum communication
and prove its capacity. 
To prepare for this, we first recall the state-channel duality,
also called source-channel duality. 
It basically states that a channel coding task can simulate a state conversion task
with the same rate, and vice versa. 

\begin{theorem}
[Source-channel duality~\cite{DHW08}] 
A channel coding task $\Phi^n \ra \Omega^k$ and 
a state conversion task $\rho^n \ra \tau^k$ for the channel $\Phi$ resp. $\Omega$
corresponding to the state $\rho$ resp. $\tau$
can simulate each other at the same rate $k/n$.
\end{theorem}

We do not state the proof here for simplicity. 
Note that it does not claim a channel resource can be converted into a state resource, and vice versa:
instead, it states that given a rate there is a channel coding task, as well as a state conversion task, 
to achieve it. 
We next introduce Csisz\'ar's key generation protocol, 
which has a source version and a channel version, 
while here we use the source version, for convenience. 

\begin{theorem}
[Key generation~\cite{AC93}] For a $n$-bit source, denoted $(A^n,B^n,C^n)$, the key generation task 
is to generate a key between $A$ and $B$ secret from $C$ allowing public communication,  
and the capacity $C_K$ is bounded by 
\be I(A:B)-\min[I(A:C),I(B:C)]\leq C_K \leq I(A:B|C). \ee 
\end{theorem}
For the channel version, the model assumes a broadcast channel $W: A\ra B \times C$ 
and the state $(A^n,B^n,C^n)$ as the yields from $W^n$,
and the capacity needs to take an optimization over the input. 
The proof is lengthy but standard.
To prove the lower bound, it relies on the packing lemma to 
compress the source size, and the covering lemma to decouple $C$.
The upper bound relies on Fannes-type continuity inequality and data 
processing inequality. 

They further showed that 
when $A\ra B\ra C$ forms Markov chain or some other conditions, 
the upper bound reduces to the lower bound, 
as a private capacity. 
$C$ can be viewed as an evesdropper.
If instead $C$ is a centralized server and known to $B$, 
then the capacity reduces to 
the conditional mutual information $I(A:B|C)$. 
That is, the role of $C$ can be complementary: 
it can be treated as a server to boost the capacity, 
or it might be an evesdropper.
In either case, $C$ must be decoupled from $AB$.

\subsection{Conditional quantum channel capacity}

We now present the quantum version, which we call \emph{conditional quantum communication}.
It is based on quantum broadcast channel $\Phi: A \ra B \otimes C$~\cite{YHD11},
but here the goal is to maximize the communication between $A$ and $B$,
treating $C$ as an assistance (or `pilot'). 
Our capacity generalizes the private capacity for degradable broadcast channels~\cite{QSW18},
and is also different from the environment-assisted capacity~\cite{SVW05,Win05}.

To compute $I(A:B|C)$ in the channel setting, we replace $A$ by its purification $R$. 
The relevant state is $\rho^{RBC}$
after the channel $\Phi_{A\to BC}$ acts on the input $\rho^A$.
The goal is not merely to establish ebits;
debits (common randomness) are also allowed. 
An ebit can simulate a debit, and an ebit is secure due to monogamy of entanglement~\cite{Terhal04},
while a debit is not secure because correlations with other systems may exist unnoticed.

\begin{theorem}
[Conditional quantum channel capacity]
For a channel $\Phi$ from system $A$ to $B$ and $C$, 
a quantum communication aiming to maximize the quantum correlation 
between $A$ and $B$ has quantum capacity 
\bea Q_K(\Phi) &=& \frac{1}{2}\max_{\rho_{A}} I(A:B|C) \\ \nonumber
&=&\frac{1}{2}\max_{\rho_{A}} [S(RC)+S(BC)-S(C)-S(E)],\eea
for $R$ as the purification system of input $\rho_{A}$ and $E$ as the environment of $\Phi$.
\end{theorem}
\begin{proof}
(Direct) To prove the lower bound of $Q_K$, 
we use the packing lemma 
(to ensure reliable decoding at $B$ assisted by $C$) 
and the covering lemma (to guarantee that $C$'s state is sufficiently 
uncorrelated with the target correlation). 
Using the EA or SA method, one can treat $k$-qubit as 2$k$ bits
and then the lower bound follows from the classical packing–covering argument. \\
(Converse) To prove the upper bound of $Q_K$, 
one considers the debit distribution between $A$ and $B$ but with $C$ decoupled.
With similar method with the classical case based on the continuity of entropy,
data-processing inequality and subadditivity, 
the upper bound follows.
Note that debit distribution can be achieved by ebit distribution,
while the mutual information of an ebit is twice of a debit.
\end{proof}

This model differs from the EA model, 
wherein $C$ serves as the assisted entanglement. 
If $C$ belongs to the input instead of the output of the channel, 
from $I(R:B|C)=I(R:BC)-I(R:C)\leq I(R:BC)$, 
the capacity would reduce to the EA capacity.
For our model, $C$ does not belong to the purification of the input, 
and if the decoder can use it, 
it is possible to boost the capacity since $I(R:B|C)$
can be larger than $I(R:B)$.
It is also not a broadcast channel as here the side channel 
$A\ra C$ is treated as an assistance to the main channel $A\ra B$. 
We could also obtain a conditional classical capacity as twice of $Q_K$.

Furthermore, as in Csisz\'ar's model~\cite{AC93},
if $C$ is the evesdropper (say, the environment $E$ of the channel),
and when $A \ra B \ra C$ is a Markov chain, 
or the channel is degradable or has independent components,
the capacity becomes 
$\max_A [I(A:B)-I(A:E)]$, which is the optimal coherent information $\max_A I_c(A\rangle B)$.
One can also treat the conditional system $C$ and 
the environment $E$ as independent, 
and decoupling both would lead to a capacity $\max_A I_c(A\rangle BC)$.

It is well established that for a channel $\Phi$, 
$Q(\Phi)\leq P(\Phi) \leq C(\Phi)$, for the three as the quantum capacity,
private capacity, and classical capacity of $\Phi$, respectively~\cite{Wil17}. 
Neither one has `single-letter' formula, 
and moreover, $P(\Phi)$ is not even additive~\cite{LWZG09}.
Therefore, one choice to avoid such dilemma
is to employ a new model to define privacy of quantum channel.
Our model offers such a choice, in that 
$C$ plays the role of evesdropper while at the same time 
there is still an environment $E$. 
Only $C$ needs to be decoupled since $E$ is not useful for the evesdropper.
After all, assuming the evesdropper to fully control the environment is strong.
Then the capacity $Q_K$ is already private,
and the classical version is $2Q_K$.
However, $Q_K$ is not directly comparable to capacity of a usual channel 
since it is defined for broadcast channels.
It is also a bit subtle to realize quantum broadcast channels in practice.

We now check which states from the previous section are useful for our model. 
For $\psi_1$, $C$ can perform entanglement swapping and hold the classical outcome,
while $A$ sends the other half of its ebit to $B$.
The states $\psi_2$, $\psi_3$, and $\psi_4$ (GHZ) can be used to generate debits;
for example, tracing out $C$ from GHZ yields a debit between $A$ and $B$.
If $C$ is measured in the dual basis and the outcome is broadcast,
an ebit between $A$ and $B$ can be generated. 
The states $\psi_5$ and $\psi_6$ do not work since $C$ is not decoupled.

The QCMI has also been used for entanglement measure, notably, 
the squashed entanglement~\cite{CW03}.
However, the usage is distinct. 
As an entanglement measure, the squashed entanglement is superadditive, 
a feature that is undesirable for capacity measure. 
This actually motivated our search for a protocol beyond entanglement distribution,
and we use QCMI as a measure of quantum correlation.

\begin{table*}[t!]
    \centering
    \caption{Quantum protocol `family tree'.}
    \label{tab:family_tree}
    \setlength{\tabcolsep}{4pt}
    \renewcommand{\arraystretch}{1.5}\footnotesize
    \begin{tabular}{|c|c|c|c|}
        \hline
        & \textbf{Merging} & \textbf{Distillation} & \textbf{Channel} \\
        \hline
        \textbf{Unassisted} 
        & $\rho^{A|B} + S(A|B)e + I(A:E)c \ge \rho^{AB}$
        & $\rho +  I(A:E)c \ge S(A|B)e$
        & $\Phi \ge S(A|B)q$ \\
        \hline
        \textbf{EA} 
        & $\rho^{A|B} + \frac12 I(A:E)q \ge \frac12 I(A:B)e + \rho^{AB}$
        & $\rho + \frac12 I(A:E)q \ge \frac12 I(A:B)e$
        & $\Phi + \frac12 I(A:E)e \ge \frac12 I(A:B)q$ \\
        \hline
        \textbf{Conditional} 
        & $\rho^{A|_CB}  + \frac12 I(A:E|C)q \ge \frac12 I(A:B|C)e +\rho^{ABC}$
        & $\rho + \frac12 I(A:E|C)q \ge \frac12 I(A:B|C)e$
        & $\Phi + \frac12 I(A:E|C)e \ge \frac12 I(A:B|C)q$ \\
        \hline
    \end{tabular}
\end{table*}

\section{Family of coding models}\label{sec:relations}

In this section, we study the family of coding models and analyze how our new models fit into 
the family tree. 
A convenient way to express models is via resource inequality~\cite{DHW08}.
All asymptotic rates are per channel use (or per copy of the resource state). 
We consider three basic noiseless resources: 
noiseless channel $q$, classical channel $c$, and ebit $e$.
The following noiseless resource inequalities hold:
\begin{align}
e + 2c &\ge q, \quad \text{(teleportation)} \\
e + q &\ge 2c, \quad \text{(superdense coding)} \\
q &\ge e. \quad \text{(entanglement distribution)}
\end{align}
Also we divide models into three classes, shown in Table~\ref{tab:family_tree},
for source coding (merging), state conversion (distillation),
and communication over a channel. 
For simplicity, we use $\ge$ to mean achievability,
exact or asymptotic.

First, consider the merging class. 
We denote a bipartite state before and after the merge as 
$\rho^{A|B}$ and $\rho^{AB}$, respectively, and label its purification system as $E$.
The coherent state merging, also known as the fully quantum Slepian–Wolf protocol, gives
\begin{equation}
\rho^{A|B}
+ \frac12 I(A:E)\,q
\;\ge\; \frac12 I(A:B)\,e \;+\; \rho^{AB}.
\end{equation}
This is an EA protocol. 
Replacing $q$ by $c$ yields the standard state merging:
\begin{equation}
\rho^{A|B} + I(A:E)\,c
\;\ge\; S(A|B)\,e + \rho^{AB}.
\end{equation}
The EA version distills ebits, while the unassisted version may consume 
ebits depending on the sign of $S(A|B)$.

With conditioning, state redistribution enacts the merge $\rho^{AC|B} \to \rho^{C|AB}$~\cite{DY08}.
Here, the side system $C$ is accessible to both $A$ and $B$, as in our conditional communication model.
In other words, it also requires the merge $\rho^{A|BC} \to \rho^{ABC}$.
We denote this conditional merge as $\rho^{A|_CB} \to \rho^{ABC}$.
The total ebit cost (or gain) is $\frac12[S(A|B)+S(A|BC)]$.
The conditional merge inequality is
\begin{equation}
\rho^{A|_CB} + I(A:E|C)\,c
\;\ge\; \frac12[S(A|B)+S(A|BC)]\,e + \rho^{ABC}.
\end{equation}
Its coherent version is
\begin{equation}
\rho^{A|_CB} + \frac12 I(A:E|C)\,q
\;\ge\; \frac12 I(A:B|C)\,e + \rho^{ABC}.
\end{equation}

Merging can be used to distill ebits, yielding the EA ebit distillation (the ``mother'' protocol):
\begin{equation}
\rho + \frac12 I(A:E)\,q \;\ge\; \frac12 I(A:B)\,e.
\end{equation}
Its dual is the EA quantum communication over a channel (the ``father'' protocol):
\begin{equation}
\Phi + \frac12 I(A:E)\,e 
\;\ge\; \frac12 I(A:B)\,q.
\end{equation}
Replacing $q$ by $c$ gives $\Phi + S(A)\,e \ge I(A:B)\,c$.
Its dual is noisy superdense coding:
$\rho + S(A)\,q \ge I(A:B)\,c$.
Using $q \ge e$, we get $\Phi \ge I_c(A\rangle B)\,e$,
and its dual is noisy teleportation:
$\rho + I(A:B)\,c \ge I_c(A\rangle B)\,q$,
modulo classical communication.

\section{Example channels}\label{sec:examples}

In this section, we construct example channels to illustrate the usage of QCMI. 
First, we note that if a broadcast channel $A \to B \otimes C$ is constructed from 
$A \to B$ and $B \to C$, or $A \to B$ and $A \to C$, 
then it suffers from a \emph{bottleneck} effect, reducing the capacity to the minimum 
of the two bipartite channels.  
To benefit from the interference effect within a broadcast channel, 
one needs genuinely multipartite channels.

Few explicit forms of broadcast channels are known. 
First, we analyze a setting where a broadcast channel can arise.
In quantum computing, logical qubits are encoded into
many physical qubits.
Correlated noise acting at the physical level
can induce effective correlations between logical qubits. 
Instead of treating this as purely detrimental, one can
sacrifice one logical qubit as an auxiliary system to boost the
accessible information of another logical qubit, 
via an \emph{interference} effect.

Suppose two logical qubits $L_1$ and $L_2$ undergo a joint noisy
evolution described by a two-qubit channel $\mathcal{N}: L_1 L_2 \to B C$,
where $B$ is the output of $L_1$ and $C$ is the output of $L_2$.
Because the noise is correlated, $C$ carries diagnostic information
about the errors that occurred on $B$. 
As an illustration, consider a two-qubit channel modeling
correlated bit-flip and phase-flip errors, with Kraus operators
\begin{align}
K_0 &= \sqrt{1-p-q}\, \I \otimes \I, \\
K_1 &= \sqrt{p}\, X \otimes X, \\
K_2 &= \sqrt{q}\, Z \otimes Z,
\end{align}
where $p,q \ge 0$ and $p+q \le 1$. 
For a maximally mixed input $\rho_{L_1L_2} = \I/4$, 
we find $I(R:B|C) = 2$ and $I(R:B) = 2 - H_3(1-p-q,p,q)$.
Thus $C$ provides perfect error diagnosis: it tells the decoder whether an $XX$ or $ZZ$
error (or neither) has occurred, unlocking the full two bits
of quantum correlations that are hidden when $C$ is discarded.

Next, we illustrate a scheme using ``environment splitting'' to construct broadcast channels from standard ones. 
Let $\mathcal{N}_1: A \to B$ be a noisy quantum channel with isometric extension
$V_1: \mathcal{H}_A \to \mathcal{H}_B \otimes \mathcal{H}_E$, so that
$\mathcal{N}_1(\rho) = \operatorname{Tr}_E (V_1 \rho V_1^\dagger)$.
Introduce a second isometry $V_2: \mathcal{H}_E \to \mathcal{H}_C \otimes \mathcal{H}_{E'}$
that splits the environment into an auxiliary system $C$ (given to the receiver)
and a discarded remainder $E'$. The overall isometry is
\begin{equation}
U = (\I_B \otimes V_2) \circ V_1 : \mathcal{H}_A \to \mathcal{H}_B \otimes \mathcal{H}_C \otimes \mathcal{H}_{E'},
\end{equation}
yielding the broadcast channel
\begin{equation}
\mathcal{N}_{\text{bc}}(\rho) = \operatorname{Tr}_{E'} (U \rho U^\dagger), \qquad
\mathcal{N}_{\text{bc}} : A \to BC.
\end{equation}

Let $R$ purify the input state $\rho_A$. The global state
$\ket{\Psi}_{RBCE'} = (I_R \otimes U)\ket{\psi}_{RA}$ is pure.
Discarding $E'$ yields a mixed state on $RBC$, and the conditional mutual
information is
\begin{equation}
I(R:B|C) = S(RC) + S(BC) - S(C) - S(E').
\end{equation}
The unconditional QMI of the original channel is
\begin{equation}
I(R:B) = S(R) + S(B) - S(RB).
\end{equation}
We define the \emph{gain} as
\begin{equation}
\Delta \equiv I(R:B|C) - I(R:B).
\end{equation}
For capacity, each quantity above should be optimized over the input,
but here our goal is to illustrate the possible gap.

We consider a rank-2 qubit channel with Kraus operators 
\begin{equation}
K_0 = \begin{pmatrix} \sqrt{1-p} & 0 \\ 0 & \sqrt{1-q} \end{pmatrix}, \qquad
K_1 = \begin{pmatrix} 0 & \sqrt{q} \\ \sqrt{p} & 0 \end{pmatrix},
\end{equation}
parameterized by $p,q \in [0,1]$.
Special cases include the amplitude damping (AD) channel ($q=0, p=\gamma$),
the dephasing channel (De.) ($p=q=\lambda$), and the bit-flip channel ($p=q=\lambda$).

We now consider four splitting combinations, shown in Table~\ref{tab:numerical}.
We take the input to be maximally mixed, $\rho_A = \I/2$, 
which is near-optimal for the channels considered.
The gain functions for the four cases are:
\begin{align}
\text{AD + AD:} \quad
\Delta &= H_2(\gamma\lambda/2)
- H_2(\gamma(1-\lambda)/2) \notag \\
&\quad + 2 - H_2((1-\gamma)/2)
- \tfrac12 H_2(\gamma), \\[4pt]
\text{De. + De.:} \quad
\Delta &= 2H_2(p) - H_2(q), \\[4pt]
\text{De. + AD:} \quad
\Delta &= H_2(p\lambda) - H_2(p(1-\lambda)) + H_2(p), \\[4pt]
\text{AD + De.:} \quad
\Delta &= H_2(\gamma/2) - H_2(q) + 2 \notag \\
&\quad - H_2((1-\gamma)/2)
- \tfrac12 H_2(\gamma).
\end{align}

\begin{table}[t!]
\centering
\caption{Numerical values of $I(R:B)$, $I(R:B|C)$, and the gain $\Delta$
for the four splitting combinations.}
\label{tab:numerical}
\begin{tabular}{|l|ccc|}
\hline
Case (parameters) & $I(R:B)$ & $I(R:B|C)$ & $\Delta$ \\ \hline 
AD + AD ($\gamma=0.3, \lambda=0.2$)   & 1.374 & 1.504 & 0.130 \\
De. + De. ($p=0.1, q=0.2$)         & 1.531 & 1.748 & 0.217 \\
De. + AD ($p=0.2, \lambda=0.3$)     & 1.278 & 1.743 & 0.465 \\
AD + De. ($\gamma=0.3, q=0.2$)      & 1.374 & 1.887 & 0.513 \\ \hline
\end{tabular}
\end{table}

Physically, the broadcast channel constructed above may appear artificial, 
yet there may be practical settings for them. 
For instance, in quantum optics, the AD channel can be induced by a beam splitter,
which splits one mode into two. 
An additional beam splitter can generate a third mode, which may serve as the environmental mode.
This suggests that if some information from the environment can be used, 
it can benefit the coding rate.
This could also be relevant in quantum metrology, 
where a `pilot' signal is used to help the main signal detect some unknown process.

\section{Conclusion}\label{sec:conclusion}

In this work, we have introduced a new quantum channel capacity,
the conditional quantum capacity,
expressed by the quantum conditional mutual information (QCMI). 
This result bridges a long-standing gap between the static interpretation of QCMI, 
established through state redistribution, 
and its dynamic role as a genuine channel capacity. 
By treating the third system \(C\) as an auxiliary output that assists decoding, 
we have shown that the optimal rate for establishing quantum correlation 
between the input and output is given by half the QCMI. 
This positions QCMI as a fundamental quantity in the whole family of quantum coding protocols.

Since QCMI is intimately related to privacy and secrecy, 
our framework has direct applications in quantum key distribution, 
entanglement swapping, and quantum repeater networks. 
In a multi-node network, 
the conditional system \(C\) can be interpreted as the collective information 
held by intermediate nodes or environmental monitors, 
and the conditional capacity characterizes the maximum rate at which 
end nodes can establish secure correlations in the presence of such auxiliary information. 

Our work also offers new methods for code design. 
Unlike conventional quantum error correction, 
which typically treats noise as an independent and adversarial process, 
the conditional communication framework explicitly exploits the structure of correlated noise. 
When the environment or an auxiliary system carries diagnostic information about 
errors-for instance, through a broadcast channel with correlated bit-flip 
and phase-flip errors-that information can be used to boost the communication rate 
beyond what is achievable by discarding the auxiliary output. 
This suggests that QCMI-based coding may outperform standard error-correcting codes,
e.g., with higher thresholds,
in settings where noise correlations are present, 
such as in multi-qubit gates, crosstalk-dominated architectures, and photonic networks.

\section{Acknowledgement}

This work has been funded by
the National Natural Science Foundation of China under Grants
12447101 and 12105343.
Discussion with D. Yang at an early stage of this work is greatly acknowledged.

\end{spacing}

\bibliography{ext}{}
\bibliographystyle{elsarticle-num}

\end{document}